\documentclass[conference]{IEEEtran}
\usepackage{epsf}
\usepackage{graphicx}
\usepackage{amsmath,bm}
\usepackage{amssymb}
\usepackage{epsfig,latexsym,amsmath,epsf,amssymb,amsfonts}
\usepackage{verbatim}
\usepackage{placeins}
\usepackage[hang]{subfigure}
\usepackage{epstopdf}

\usepackage{amsthm}
\newtheorem{lemma}{Lemma}
\newtheorem{corollary}{Corollary}
\newcommand{\ignore}[1]{}

\begin{document}
\title{Generalized Instantly Decodable Network Coding for Relay-Assisted Networks}
\author{\large Adel M. Elmahdy$^1$, Sameh Sorour$^2$, and Karim G. Seddik$^3$\\ [.1in]
\normalsize  \begin{tabular}{c}
$^1$Electrical Engineering Department, Alexandria University, Alexandria 21544, Egypt.\\
$^2$Computer, Electrical and Mathematical Sciences and Engineering (CEMSE) Division,\\
King Abdullah University of Science and Technology (KAUST), Thuwal, Kingdom of Saudi Arabia.\\
$^3$Electronics Engineering Department, American University in Cairo, AUC Avenue, New Cairo 11835, Egypt.\\
Email: adel.elmahdy@ieee.org, sameh.sorour@kaust.edu.sa, kseddik@aucegypt.edu
\end{tabular}
}
\maketitle

\begin{abstract}
In this paper, we investigate the problem of minimizing the frame completion delay for Instantly Decodable Network Coding (IDNC) in relay-assisted wireless multicast networks. We first propose a packet recovery algorithm in the single relay topology which employs \textit{generalized IDNC} instead of \textit{strict IDNC} previously proposed  in the literature for the same relay-assisted topology. This use of generalized IDNC is supported by showing that it is a super-set of the strict IDNC scheme, and thus can generate coding combinations that are at least as efficient as strict IDNC in reducing the average completion delay. We then extend our study to the multiple relay topology and propose a joint generalized IDNC and relay selection algorithm. This proposed algorithm benefits from the reception diversity of the multiple relays to further reduce the average completion delay in the network. Simulation results show that our proposed solutions achieve much better performance compared to previous solutions in the literature.

\end{abstract}


\section{Introduction}
Since the introduction of \textit{Network Coding (NC)} in the seminal paper by Ahlswede \textit{et al.} \cite{J:NetInfFlow}, it has been widely employed to improve the transmission efficiency in many network settings (\cite{C:ONC2,J:RNC2} and ref. therein). Moreover, several works focused on using \emph{opportunistic network coding (ONC)} to enhance the packet recovery process in point-to-multipoint (PMP) wireless broadcast erasure channels (\cite{4476183,Sadeghi2010} and ref. therein). One of these ONC schemes that has been widely considered in this context is \emph{Instantly Decodable Network Coding (IDNC)} due to its many desirable properties such as the instant decoding of the received packets at the selected set of receivers, simple XOR encoding and decoding, and the needlessness for decoding buffers at the receiver side to store the coded packets. These properties allow simple and power efficient designs of IDNC-capable receivers, which are more suitable for hand-held mobile terminal.

Recently, NC was extended to enhance packet recovery in \emph{relay assisted networks (RANs)} \cite{fan2009reliable,J:LuXiaRas}. This extension responded to the consideration of \textit{relay nodes (RNs)} in the standardization process of next generation mobile broadband communication systems such as 3GPP LTE-Advanced \cite{J:Relay1}, \cite{J:Relay2} and \cite{4GLte}. Deploying relay nodes is a smart solution to increase the capacity or to extend the cell coverage area by enhancing the coverage at the cell edge where the receivers suffer from low Signal to Interference plus Noise Ratio (SINR). The problem of minimizing the number of recovery transmissions (a.k.a. completion delay) using IDNC was considered in the framework of one decode-and-forward relay-assisted broadcast networks \cite{J:LuXiaRas}. The proposed IDNC-based solution was shown to achieve a better completion delay and thus transmission efficiency, as compared to IDNC with no relay case as well as automatic repeat-request (ARQ).

Even though the authors of \cite{J:LuXiaRas} considered a \emph{strict IDNC (S-IDNC)} approach when solving the completion delay problem in RANs, the senders are constrained to encode packets that include at most one lost packet for each receiver in this IDNC approach. In \cite{sorour10minimum}, it has been shown that this S-IDNC approach is too restrictive in enhancing the decoding delay of the different receivers. The authors thus proposed a more relaxed IDNC approach which allowed any packet encoding at the sender and forced each of the receivers to discard any coded packets combining two or more non-received source packets at that receivers. This \emph{generalized IDNC (G-IDNC)} approach was shown to provide much more coding flexibility at the sender and thus achieved much better completion delay performance.

This paper extends the works in \cite{fan2009reliable,J:LuXiaRas} in two directions. First, we propose the use of G-IDNC as a more efficient NC scheme to further reduce the completion delay in the one-relay setting. Although the G-IDNC approach was proven to outperform S-IDNC in decoding delay only \cite{sorour10minimum}, the coding flexibility of G-IDNC and its previous use for completion delay reduction in simple PMP networks \cite{C:BCDMin,J:CDMin} can be used to outperform the completion delay achieved by S-IDNC in the one-relay scenario. We then extend our study to the case of multiple-relay scenario, in which joint decisions on the best encoding packet and best transmitting relay should be made. We thus propose a solution for this joint decision problem, which is based on the same methodology to minimize the completion delay in PMP networks with one relay node. We finally compare this proposed solution to other famous approaches in the PMP NC literature.

\section{System Model}\label{intro}
The proposed system consists of a base station (BS) that is required to deliver a frame $\mathcal{N}$ of $N$ source packets to a set $\mathcal{M}$ of $M$ terminal nodes (TNs) over erasure channels. Each of the receivers is interested in a subset or all of the packets of $\mathcal{N}$, which defines a generic multicast scenario. The packets requested by receiver $i$ are referred to as the primary packets, while the undesired packets are referred to as the secondary packets. In addition, there is a set $\mathcal{R}$ of $R$ decode-and-forward RNs to improve the coverage of the BS. An example of such network with 3 RNs is shown in Fig. \ref{fig:SystemModel}. Similar to the assumptions and constraints of \cite{J:LuXiaRas}, we assume that the channels between the BS and RNs are better than those between the BS and the TNs. We also assume that the channel between the RNs and the TNs are better than those between the BS and the TNs. RNs operate in half-duplex mode and cannot transmit simultaneously with the BS, even if they have received all the packets in $\mathcal{N}$.

\begin{figure}
\centering
\includegraphics[width=0.5\linewidth]{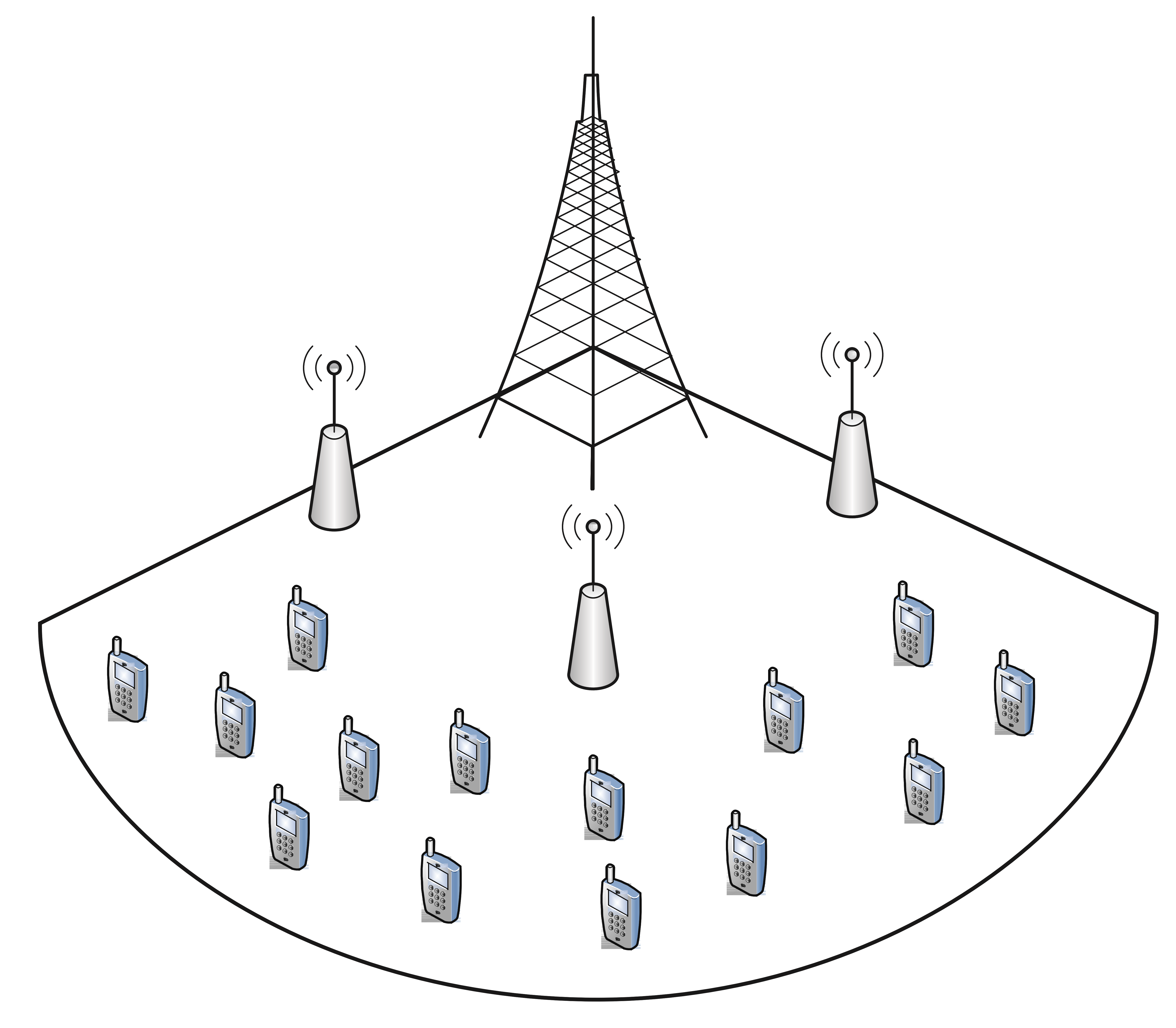}
\caption{Proposed system model for 3 relay nodes}
\label{fig:SystemModel}
\vspace{-.2in}
\end{figure}

In the \textit{initial transmission phase}, the BS sequentially sends all the source packets of the frame. These transmissions are subject to possible erasures at the RNs and TNs. For each received packet, each RN/TN feeds back a positive acknowledgement (ACK) to the BS. RNs also store the feedback sent by the TNs. When the initial transmission phase comes to an end, the BS can attribute 3 sets of packets to each TN/RN:
\begin{itemize}
\item The \textit{Has} set ($\mathcal{H}_i$): The set of primary and secondary packets that are successfully received by receiving node $i$.
\item The \textit{Lacks} set ($\mathcal{L}_i$): The set of all packets that are not received by receiving node $i$. Note that $\mathcal{N}$ = $\mathcal{H}_i \cup \mathcal{L}_i$.
\item The \textit{Wants} set ($\mathcal{W}_i$): The set of wanted packets by $i$ that are not received by it. Note that $\mathcal{W}_i \subseteq \mathcal{L}_i$.
\end{itemize}
Note that the Wants sets of RNs are always empty as they do not want any packets. The sender stores this information in a \emph{state feedback matrix (SFM)} $\mathbf{F} =
\left[f_{ij}\right],~\forall~i\in\mathcal{M},j\in\mathcal{N}$, such that $f_{ij} = 0$ if $j \in \mathcal{H}_i$, $f_{ij} =1$ if $j \in \mathcal{W}_i$, and $f_{ij} = -1$ if $j \in \mathcal{L}\setminus\mathcal{W}_i$.

Afterwards, a recovery phase begins in which both the BS and RNs should collaborate in order to deliver the missing packets of all TNs in the smallest number of transmissions. Both BS and RNs employ IDNC packet that is able to serve many TNs simultaneously in each recovery transmission. Note that each recovery transmission from the BS or a RN is still subject to erasures. Thus, the TNs (and RNs if the BS is transmitting) should respond to each received recovery transmission with an ACK, which is used to update the SFM for subsequent transmissions. The method with which these IDNC coded packets are generated to target specific TNs (and possibly RNs) and the algorithms for BS/RN selection for each recovery transmission will be described in details in Sections \ref{IDNCGraph} and \ref{Algorithm}, respectively.

\ignore{%%%%%%%%%%%%%%%%%%%%%%%%%%%%%%%%%%%%%%%%%%%%%%%%%%%%%%
In the first step, the BS transmits recovery coded packets which are instantly decodable to a selected group of TNs. These packets can also be instantly decodable at the RNs, thus supplying them with additional source packets. The method these coded packets is generated with to target specific TNs and RNs will be described in the next section. This process continues until either the Wants sets of all TNs are empty or each RN receives all the packets that are still wanted by any of the TNs. In the second step, one of the RNs is selected per transmission to send a coded packet to the TNs. This process continues till all the Wants sets of all TNs become empty. Note that in both steps of the recovery phase, the TNs (and also RNs in Step 1) then respond to each received (non-erased) recovery transmission with an ACK, which is used to update the SFM for subsequent transmissions.
}%%%%%%%%%%%%%%%%%%%%%%%%%%%%%%%%%%%%%%%%%%%%%%%%%%%%%%%%%%%%5555

\section{G-IDNC Graph}\label{IDNCGraph}
The IDNC graph \cite{sorour10minimum, C:BCDMin} provides a framework to determine the packets that can be combined together and simultaneously recovered by any given set of TNs using the same transmission. This graph $\mathcal{G}$ is constructed by first generating a vertex $v_{ij}$ in $\mathcal{G}$ for each packet $j \in \mathcal{L}_i$, and for all users. Two vertices $v_{ij}$ and $v_{kl}$ in $\mathcal{G}$ are adjacent if one of the following conditions is
true:
\begin{itemize}
\item C1: $j = l$ $\Rightarrow$ The two vertices represent the loss of the same packet $j$ by two different TNs $i$ and $k$.
\item C2: $j\in \mathcal{H}_k$ and $l \in \mathcal{H}_i$ $\Rightarrow$ The requested packet of each vertex is in the Has set of the TN of the other vertex.
\end{itemize}
Consequently, each edge between two vertices $v_{ij}$ and $v_{kl}$ in the graph means that the missing packets $j$ and $l$ at TNs $i$ and $k$ can be served simultaneously and instantly by sending either packet $j$ or $j\oplus l$ if $j=l$ or $j\neq l$, respectively. This property extends from two adjacent vertices to all cliques in the graph. A clique in a graph is a subset of that graph whose vertices are all adjacent to one another. Thus, each clique in  $\mathcal{G}$ defines a packet combination that can instantly serve all the users inducing this clique's vertices.

We can classify the vertices of a graph into two layers:
\begin{itemize}
\item Primary graph $\mathcal{G}_\rho$: It includes all the vertices from the Wants sets of all users.
\item Secondary graph $\mathcal{G}_\sigma$: It includes all the vertices that are not in the Wants set of any user.
\end{itemize}
For any transmission, we first select a clique $\kappa_\rho$ from the primary graph, which will target a set of TNs $\mathcal{T}_{\kappa_\rho}\subseteq\mathcal{M}$. For the non-targeted TNs and most importantly the RNs (i.e. $(\mathcal{M}\setminus\mathcal{T}_{\kappa_\rho}) \cup \mathcal{R}$), the adjacent secondary subgraph to $\kappa_\rho$ can be used to deliver unwanted packets to them without violating the instant decodability of the packets $\kappa_\rho$ at the TNs in $\mathcal{T}_{\kappa_\rho}$. For non-targeted TNs, this step enlarges their Has sets and thus increases the chances of having more coding opportunities in subsequent transmissions, according to condition C2. Furthermore, the RNs can receive more source packets, which increases their capability to create and send more efficient coding combinations in subsequent transmissions.

Fig. \ref{fig:Graph} depicts an example of a state feedback matrix and the corresponding IDNC graph.

\begin{figure}
\centering
\includegraphics[scale=0.38]{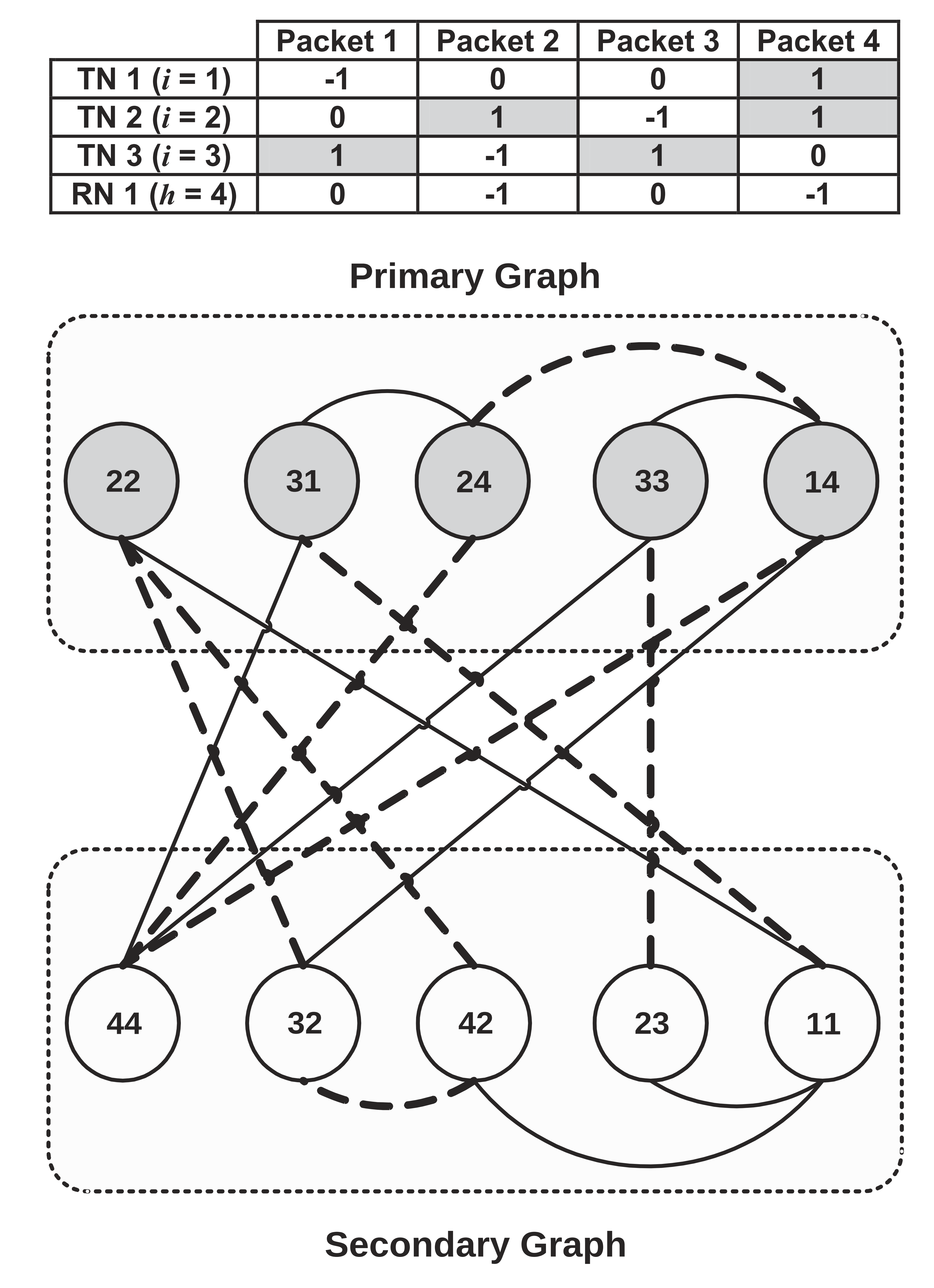}
\caption{An example of a state feedback matrix and the corresponding G-IDNC graph in a one-RN model. Dashed and solid lines represent the edges generated by C1 and C2, respectively.}
\label{fig:Graph}
\end{figure}

\section{Why G-IDNC instead of S-IDNC}\label{PerformanceComparison}
In this section, we compare the performance of G-IDNC against S-IDNC and show the ability of the former to achieve a better step towards a lower completion delay, given any feedback matrix and for any scheduling approach.

We start our analysis by proving the following lemma.
\begin{lemma} \label{T1}
For any feedback matrix, all coding combinations of any S-IDNC scheme can be formulated as a graph $\mathcal{G}_s$ that is a subgraph of the G-IDNC graph $\mathcal{G}$.
\end{lemma}

\begin{proof} \label{P1}
For any feedback matrix, we can generate a graph $\mathcal{G}_s$ for S-IDNC packet transmission, which will have the same vertices of the corresponding G-IDNC graph $\mathcal{G}$ for the same feedback matrix. Since packet combination possibilities are represented by the edges in the graph, we can enforce the S-IDNC constraint on the vertex connectivity conditions. Defining $\mathcal{M}_j$ and $\mathcal{M}_l$ as the sets of receivers that lack packets $j$ and $l$, respectively, we can express the S-IDNC conditions to serve two vertices simultaneously as follows:
\begin{itemize}
\item $\widetilde{\mbox{C}}1$: $j=l$.
\item $\widetilde{\mbox{C}}2$: $j \in \mathcal{H}_k\ \mbox{and}\ l \in \mathcal{H}_i\ \mbox{and}\ \mathcal{M}_j\cap \mathcal{M}_l = \emptyset$.
\end{itemize}
Indeed, for condition $\widetilde{\mbox{C}}1$ ($j = l$), the edges do not represent coding combinations since both end vertices of the edge request the same packet. Consequently, this edge on its own represents the transmission of the source packet $j$, which can never be a non-instantly decodable for any TN, and thus does not violate the coding constraints in S-IDNC. As for condition $\widetilde{\mbox{C}}2$, it does indeed involve a packet combination $j \oplus l$ but again the last constraint $\mathcal{M}_j\cap \mathcal{M}_l = \emptyset$ guarantees that no coding of two packets missed by the same receiver can be done.

Comparing these conditions to those in Section \ref{IDNCGraph} (C1 and C2), it is obvious that both C1 and $\widetilde{\mbox{C}}1$ are equivalent. Thus, all the adjacent vertices using condition C1 in $\mathcal{G}_s$ will also be adjacent in $\mathcal{G}$. On the other hand, C2 is a relaxed condition of $\widetilde{\mbox{C}}2$, and thus any two vertices satisfying $\widetilde{\mbox{C}}2$ also satisfy C2 but not vice versa. Consequently, all edges generated in $\mathcal{G}_s$ using $\widetilde{\mbox{C}}2$ will also be generated in $\mathcal{G}$ using C2. However, $\mathcal{G}$ will have extra edges between vertices that satisfy C2 but not $\widetilde{\mbox{C}}2$.

We can thus conclude that $\mathcal{G}_s$ is a subgraph of $\mathcal{G}$ for any feedback matrix.
\end{proof}
\ignore{
\begin{corollary}
The erasure-less completion delay of G-IDNC is at most as large as that of the S-IDNC scheme.
\end{corollary}
\begin{proof}
Assuming erasure-less recovery transmissions, it is known that, given the graph representation of both the G-IDNC and S-IDNC, the completion delay of each of these scheme is equivalent to the chromatic number of their corresponding complimentary graph \cite{TON09}. The complimentary graph $\mathcal{G}^c\left(\mathcal{V}^c,\mathcal{E}^c\right)$ of a graph $\mathcal{G}\left(\mathcal{V},\mathcal{E}\right)$ is defined such that $\mathcal{V}^c = \mathcal{V}$ (i.e. both graphs have the same vertex set) and $\mathcal{E}^c = \left(\mathcal{V}\times\mathcal{V}\right)\setminus \mathcal{E}$ (i.e. $\mathcal{G}^c$ has all the edges of the complete graph consisting of the vertex set $\mathcal{V}$ except those in $\mathcal{E}$)
\end{proof}
}

From \cite{J:CDMin}, it has been shown that a transmission aiming to reduce the completion delay should target the largest number of receivers with largest Wants sets and erasure probabilities, as such transmission brings the system the closest to completion by:
\begin{itemize}
\item Reducing the largest individual completion delays of worst case receivers.
\item Maximizing the resulting number of coding opportunities for subsequent transmissions.
\end{itemize}
Such strategy, termed the \emph{WoRLT strategy}, can be applied by attributing the weights $\left(\frac{\left|\mathcal{W}_i\right|}{1 - p_i}\right)^n$ to the vertices of each TN $i$ in the IDNC graph, and choosing the maximum weight clique (MWC) ($p_i$ is the erasure probability of TN $i$). For large $n$, the chosen MWCs are always the ones satisfying the WoRLT strategy.

After the initial transmission phase (i.e. at time $t=0$), both $\mathcal{G}_s$ and $\mathcal{G}$ can be generated for the corresponding feedback matrix. Lemma 1 shows that $\mathcal{G}_s$ will be a subset of $\mathcal{G}$. Consequently, the MWC $\kappa^*_s$ chosen from $\mathcal{G}_s$ will definitely be a subset of or equal to a maximal clique $\kappa_g(\kappa^*_s)$ in $\mathcal{G}$, which need not be the MWC in $\mathcal{G}$. Thus, the weight $\omega\left\{\kappa_g(\kappa^*_s)\right\}$ of $\kappa_g(\kappa^*_s)$ will be definitely smaller than or equal to that of the MWC $\kappa^*_g$ in $\mathcal{G}$. Consequently, $\omega\left\{\kappa_g^*\right\} \geq \omega\left\{\kappa_g(\kappa^*_s)\right\} \geq \omega\left\{\kappa_s^*\right\}$, and thus $\kappa_g^*$ can bring the system closer to completion compared to $\kappa_s^*$.

Assume that we send the packet corresponding to $\kappa_g^*$ at $t=0$. It is easy to see by repeating the same analysis that sending the MWC of the resulting G-IDNC graph at $t=1$ will bring the system even closer to completion. Consequently, we can infer from the previous discussion that the use of the G-IDNC graph all along the recovery phase will bring the system to completion at least as fast as the use of S-IDNC graph.

\section{Proposed Algorithm}\label{Algorithm}
\ignore{The contribution of this paper is the enhanced reduction in the expected completion delay that is achieved by G-IDNC scheme in communication networks assisted by decode-and-forward RNs. In such a network, for every recovery transmission, each transmitting node applies an algorithm, such as \textit{Maximum Weight Clique Selection Algorithm} \cite{C:ExAlgMWC} and \textit{Maximum Weight Vertex Search Algorithm} \cite{J:CDMin} that exploits the WoRLT strategy \cite{C:BCDMin}, \cite{J:CODen} to select the maximal clique that achieves the minimum completion delay for this node. WoRLT strategy simultaneously achieves the following goals:
\begin{enumerate}
\item Minimizing the distance between the Wants vector of the system and the absorbing states.
\item Maximizing the coding opportunities in the primary graph of successor states.
\end{enumerate}	}

\subsection{One-RN Scenario}
The algorithm of the recovery transmission phase for the one-RN case is performed on two steps as follows:\\
\textbf{\textit{Step 1:}}
The BS carries out the re-transmission process to the TNs and the RN as follows:
\begin{enumerate}
\item The weight of the vertices of each TN $i$ is set to be $\left(\frac{\left|\mathcal{W}_i\right|}{1 - p_{_{(BS,i)}}}\right)^n$, where $p_{_{(BS,i)}}$ is the erasure probability from the BS to TN $i$. Note that the RN has an empty Wants set and all its vertices are in the secondary graph. We thus assign a very low weight to its vertices in the secondary graph, which makes their selection the last priority. This is similar to the assumption used in \cite{J:LuXiaRas}, except that we do care about RN decodability when possible\ignore{, whereas the solution in \cite{J:LuXiaRas} does not care about it at all at this step}.
\item Select the primary MWC $\kappa_\rho$ in the IDNC primary graph. By doing so, the TNs that have the largest Wants sets and erasure probability are targeted.
\item Extract the secondary subgraph $\mathcal{G}(\kappa_\rho)$ consisting of secondary vertices each of which being adjacent to all the vertices in $\kappa_\rho$.
\item Select the MWC $\kappa_\sigma$ in $\mathcal{G}(\kappa_\rho)$. By doing so, the TNs that have the largest Wants sets but lower erasure probability of receiving a packet are targeted with unwanted information. If possible, the RN is also is targeted with extra packets.
\item Send the packet corresponding to the combined clique $\kappa^* = \kappa_\rho \cup \kappa_\sigma$.
\item Update the feedback matrix according to the received feedback and repeat this step again until either of these conditions occur:
    \begin{itemize}
    \item The Wants sets of all TNs are depleted, which means that they received all their wanted packets.
    \item The RN receives all the packets that are still in the Wants set of any of the TNs.
    \end{itemize}
    In the former condition, the BS declares frame delivery and start the transmission of subsequent frames, if any. In the latter, the system moves to Step 2.
\end{enumerate}
\textbf{\textit{Step 2:}}
The RN continues the recovery process by following the same procedures of Step 1, with the following two changes:
\begin{itemize}
\item There are no vertices belonging to the RN in the graph.
 \item The weights of the vertices of each TN $i$ become $\left(\frac{\left|\mathcal{W}_i\right|}{1 - p_{_{(RN,i)}}}\right)^n$, where $p_{_{(RN,i)}}$ is the erasure probability from the RN to TN $i$.
\end{itemize}
This step is repeated until all TNs receive their wanted packets.

\subsection{Multiple-RN Scenario} \label{sec:3-RN}
In the multiple-RN scenario, the algorithm should perform a selection of both the transmitting RN and the clique determining the coded transmission, in order to achieve the least possible completion delay in the whole network. Thus, the algorithm operates as follows:\\
\textbf{\textit{Step 1:}}
The BS sends NC recovery packets to the TNs and the RNs, similar to the one-RN case described above. Nonetheless, the algorithm moves to Step 2 when:
 \begin{equation}\label{eq:MRN}
 \bigcup_{i\in\mathcal{M}} \mathcal{W}_i \subseteq \bigcup_{h\in\mathcal{R}} \mathcal{H}_h.
  \end{equation}
  In other words, the algorithm moves to Step 2 when each packet that is still needed by any receiver is received by at least one RN. This condition clearly shortens Step 1 compared to the one-RN scenario because of the diversity of the Has sets at the different RNs.\\
\textbf{\textit{Step 2:}}
All $R$ RNs continue the recovery process together, under the management of the BS. For each recovery transmission, the chosen RN and clique for transmission are jointly selected as follows:
\begin{enumerate}
\item Each RN $h\in\mathcal{R}$ separately selects its MWC $\kappa^*(h)$ similar to the one-RN scenario, except that it employs the weight $\left(\frac{\left|\mathcal{W}_i\right|}{1 - p_{_{(h,i)}}}\right)^n$ for the vertices of TN $i$.
\item The BS selects the RN $h^* = \arg\max_{h\in\mathcal{R}} \;\omega\{\kappa^*(h)\}$ (i.e. the RN having the MWC with the highest weight) and gives permission to this RN to transmit the packet combination defined by $\kappa^*(h^*)$ in this recovery transmission.
\end{enumerate}
This step is repeated until all TNs receive their wanted packets. Note that Step 2 in this case may take longer than its one-RN version. Indeed, the condition in (\ref{eq:MRN}) makes it very unlikely that any of the $R$ RNs will have all the packets still wanted by all the TNs by the end of Step 1. Thus, each of the $R$ RNs will need to deliver its received subset of packets in separate transmissions, which may take longer than Step 2 in the one-RN case.

%%%%%%%%%%%%%%%%%%%%%%%%%%%%%%%%%%%%%%%%%%%%%%%%%%%%%%%%%%%%%%%%%%%%%%%%%%%%%%%%%%%%%%%%%%%%%%%%%%%%%%%%%%%%%%%

\subsection{Maximum Vertex Search Algorithms}
Since the MWC algorithm has high computational complexity, we propose a modification to the above algorithms by replacing the MWC algorithm with the maximum vertex search (MVS) algorithm, proposed in \cite{C:BCDMin,J:CDMin}, and having complexity of $O((M+R)^2N)$. This algorithm constructs the transmission clique by greedily choosing in each step the vertex that has both the largest weight and is connected to the largest number of vertices having high weights. The details of this algorithm can be found in \cite{C:BCDMin,J:CDMin}.

\section{Simulation Results}
\ignore{In this section, we first compare the performance of G-IDNC vs S-IDNC in terms of reducing the completion delay for both the one-RN and three-RN scenarios using the above WoRLT strategy based algorithms. We then compare the performance of the proposed G-IDNC three-RN algorithm when the WoRLT strategy is used compared to several imported famous approaches in the PMP NC literature.}

\subsection{G-IDNC vs S-IDNC}
For the one-RN and three-RN scenarios, Fig. \ref{CD} depicts the average completion delay against the number of TNs $M$ for G-IDNC and S-IDNC schemes. For each of the two schemes, the WoRLT strategy is employed for vertex weighting and both MWC and MVS are tested. The number of packets used in the simulation is $N = 30$ and each TN requires on average 80$\%$ of these packets. The packet erasure probabilities on the BS-TN, BS-RN and RN-TN channels are randomly selected from the ranges $[0.3,0.5]$, $[0.1,0.2]$ and $[0.05,0.15]$, respectively.
\begin{figure}
\centering
\includegraphics[width=1\linewidth]{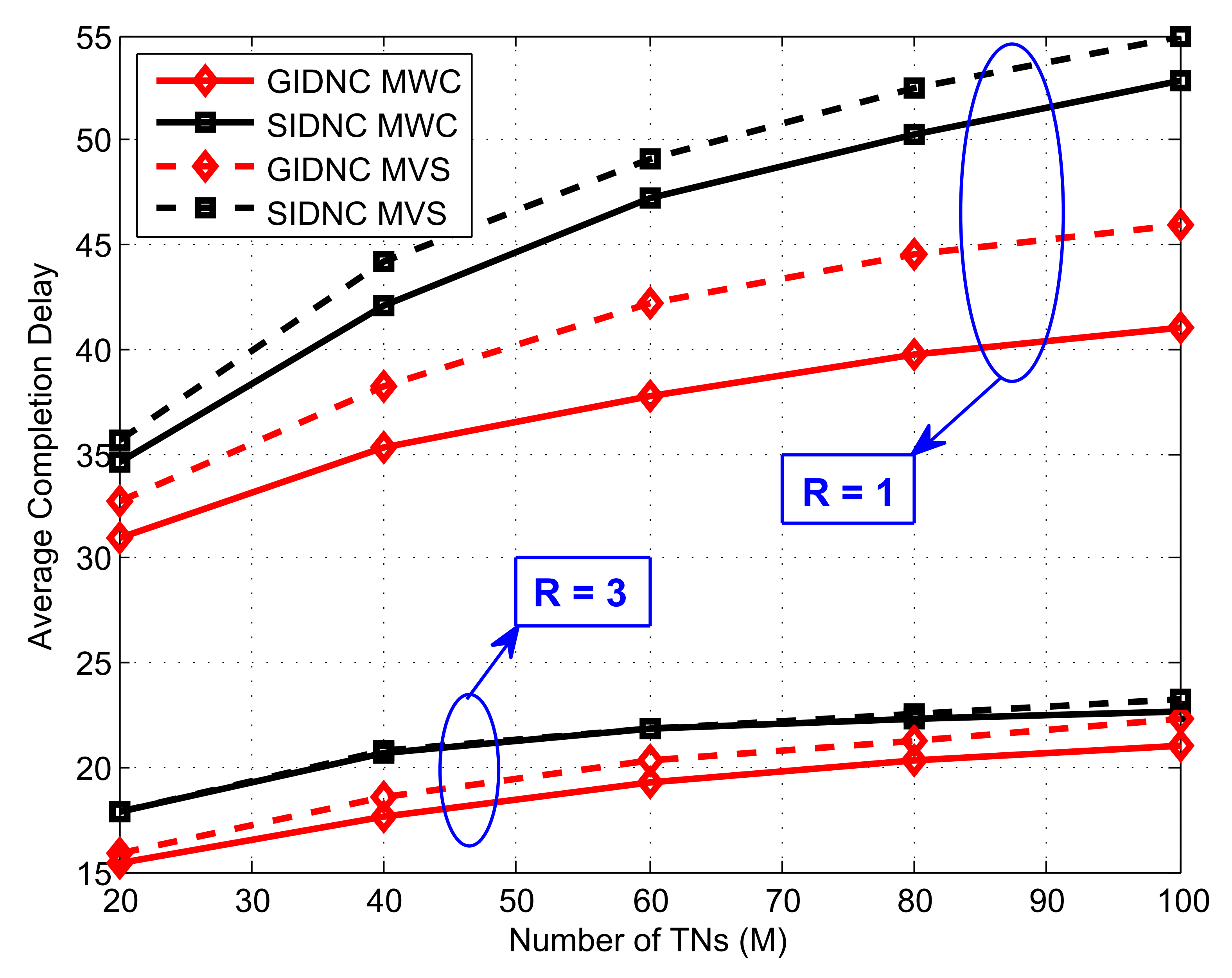}
\caption{Average completion delay vs \textit{$M$} ($N$=30, 80$\%$ average TN demand) for different scenarios (one and three RNs) and algorithms (MWC and MVS). }
\label{CD}
\end{figure}

For both the one-RN and three-RN scenarios, it is obvious that the performance of G-IDNC is superior to S-IDNC for both the MWC and MVS algorithms. Similar results have been obtained for all ranges of $N$ and TN average demand but weren't included due to space limitation. Furthermore, we can see that the three-RN case achieves a significant improvement compared to the one-RN case. This means that, compared to the one-RN algorithm, the shortening that occurs in Step 1 of the three-RN case, due to the RN reception diversity, has a much bigger impact than the lengthening of Step 2 due to the separate RN transmissions (resulting from the incomplete Has sets of individual RNs at the end of Step 1). Moreover, the results show that the MVS algorithm performance gets closer to the one of MWC algorithm at larger numbers of RNs.

\subsection{Joint Clique-RN Selection for G-IDNC}
For the G-IDNC scheme and $R=3$, Fig. \ref{Schd} compares the average completion delay achieved by the WoRLT joint clique-RN selection strategy with the following joint approaches:\\
\textbf{\textit{Approach 1}}:
Choosing the clique-RN combination that maximizes the service of the maximum clique. To find the maximum clique in each RN, the weight of each vertex is set to 1.\\
\textbf{\textit{Approach 2}}:
Choosing the clique-RN combination that serves the maximum expected number of TNs. To find the maximum expected number of TNs that can be served by each RN, the weight of each vertex $v_{ij}$ is set to $\left(1-p_{_{h,i}}\right)$.\\
\textbf{\textit{Approach 3}}:
Choosing the clique-RN combination that maximizes the service of the most wanted packets. The weight of each vertex $v_{ij}$ is set to be $P_j$, defined as the number of TNs wanting packet $j$.\\
After finding $\kappa^*(h)$ for any recovery transmission using any of the above approaches, the RN $h^*$ having $h^*  = \arg\max_{h \in \mathcal{R}} \left\{\sum_{i|v_{ij}\in \kappa^*_h} \left(1-p_{_{h,i}}\right)\right\}$ is selected by the BS to send the packet combination defined by $\kappa^*(h^*)$ in this recovery transmission. The BS thus selects the RN that has higher chance of delivering its selected clique\ignore{according to the Approach rule}.

It is obvious that WoRLT based joint clique-RN scheduling approach achieves a significant improvement in the average completion delay over all these different approaches that are common in the PMP NC literature.

%\newpage

\begin{figure}
\centering
\includegraphics[width=1\linewidth]{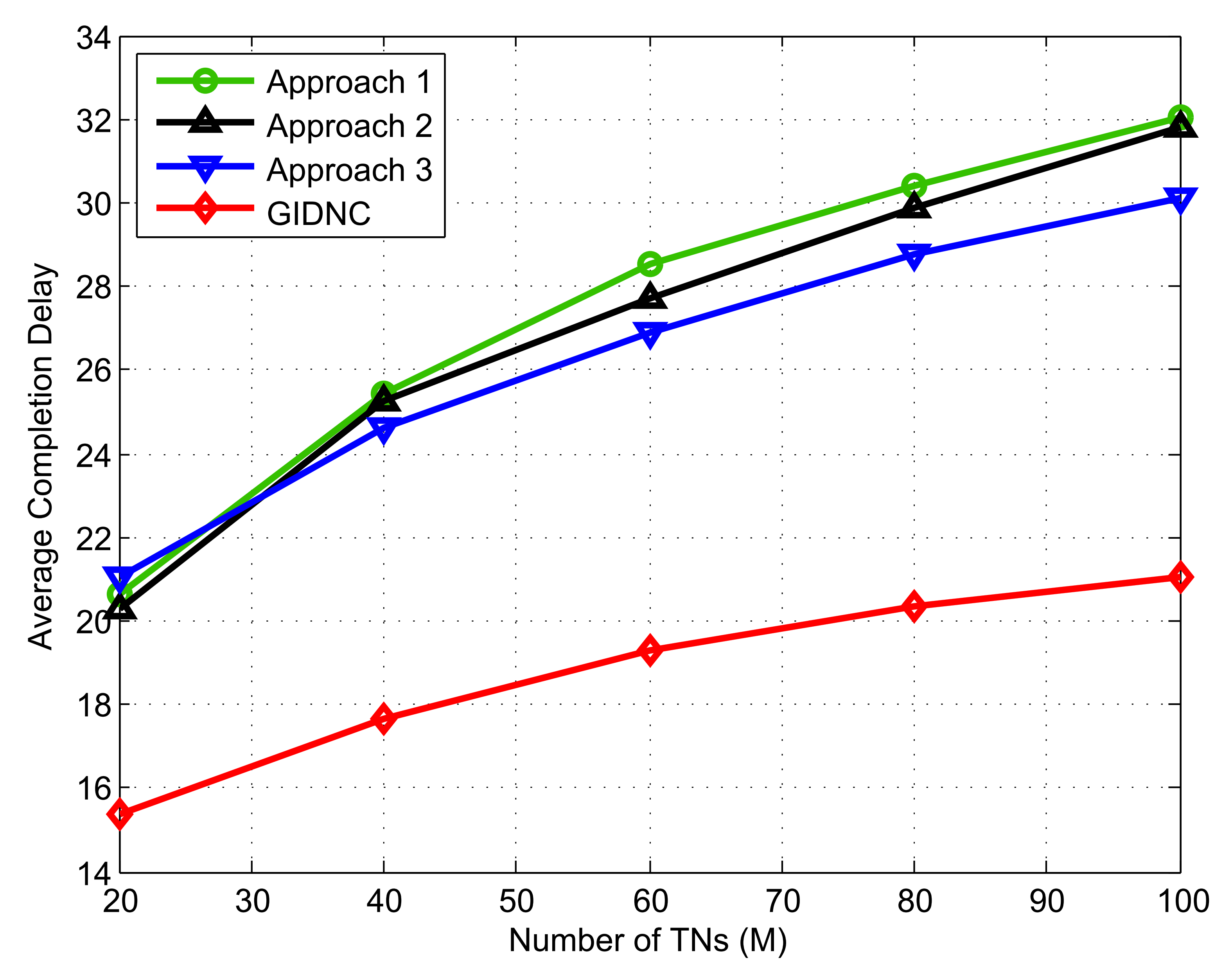}
\caption{Average completion delay vs \textit{$M$} for different scheduling schemes for the three-RN scenario ($N$=30, 80$\%$ average TN demand).}
\label{Schd}
\end{figure}

\section{Conclusion}
In this paper, we investigated the problem of minimizing the frame completion delay for instantly decodable network coding in RN-assisted wireless multicast networks. To achieve a lower completion delay, we first proposed the use of the the generalized version of the IDNC scheme to the one-RN scenario instead of the previously proposed \textit{strict} IDNC version in the literature. We supported this exchange by showing that S-IDNC provides coding combinations that are always a subset of the G-IDNC combinations for the same feedback matrix. We then designed G-IDNC based WoRLT algorithm for the one-RNs scenario, and extended it to the multiple-RN scenario. Simulation results showed that our proposed solutions achieve much better performance compared to the S-IDNC solutions in the literature. They also showed that the use of the WoRLT strategy achieves a significantly better completion delay compared to other commonly used strategies in the PMP NC literature.

%\newpage
\bibliographystyle{IEEEbib}
\linespread{1}
\bibliography{References}

\begin{thebibliography}{10}

\bibitem{J:NetInfFlow}
R.~Ahlswede, Cai N., S.-Y.R. Li, and R.W. Yeung,
\newblock ``Network information flow,''
\newblock {\em IEEE Transactions on Information Theory}, vol. 46, no. 4, pp.
  1204--1216, Jul. 2000.

\bibitem{C:ONC2}
S.~Katti, D.~Katabi, W.~Hu, H.~S. Rahul, and M.~M\'{e}dard,
\newblock ``The importance of being opportunistic: practical network coding for
  wireless environment,''
\newblock in {\em Proc. Allerton Conference on Communication, Control, and
  Computing}, Illinois, USA, Sept. 2005.

\bibitem{J:RNC2}
J-S. Park, M.~Gerla, D.S. Lun, Yunjung Yi, and M.~M\'{e}dard,
\newblock ``Codecast: a network-coding-based ad hoc multicast protocol,''
\newblock {\em IEEE Wireless Communications}, vol. 13, no. 5, pp. 76--81, Oct.
  2006.

\bibitem{4476183}
L.~Keller, E.~Drinea, and C.~Fragouli,
\newblock ``Online broadcasting with network coding,''
\newblock {\em Fourth Workshop on Network Coding, Theory and Applications
  (NetCod'08)}, Jan. 2008.

\bibitem{Sadeghi2010}
P.~Sadeghi, R.~Shams, and D.~Traskov,
\newblock ``An optimal adaptive network coding scheme for minimizing decoding
  delay in broadcast erasure channels,''
\newblock {\em EURASIP Journal of Wireless Communications and Networking}, vol.
  2010, pp. 1--14, Apr. 2010.

\bibitem{fan2009reliable}
P.~Fan, C.~Zhi, C.~Wei, and K.~Ben~Letaief,
\newblock ``Reliable relay assisted wireless multicast using network coding,''
\newblock {\em IEEE Journal on Selected Areas in Communications}, vol. 27, no.
  5, pp. 749--762, Jun. 2009.

\bibitem{J:LuXiaRas}
L.~Lu, M.~Xiao, and L.K. Rasmussen,
\newblock ``Design and analysis of relay-aided broadcast using binary network
  codes,''
\newblock {\em Journal of Communications}, vol. 6, no. 8, pp. 610--617, Nov.
  2011.

\bibitem{J:Relay1}
S.~Parkvall, E.~Dahlman, A.~Furuskar, Y.~Jading, M.~Olsson, S.~Wanstedt, and
  K.~Zangi,
\newblock ``{LTE}-advanced - evolving {LTE} towards {IMT}-advanced,''
\newblock in {\em IEEE Vehicular Technology Conference (VTC'08)}, Sep. 2008.

\bibitem{J:Relay2}
Y.Yang, H.~Hu, J.~Xu, and G.~Mao,
\newblock ``Relay technologies for wimax and {LTE}-advanced mobile systems,''
\newblock {\em IEEE Communications Magazine}, vol. 47, no. 10, pp. 100--105,
  Oct. 2009.

\bibitem{4GLte}
E.~Dahlman, S.~Parkvall, and Sk\"{o}ld J.,
\newblock {\em {4G}: {LTE}-Advanced for Mobile Broadband},
\newblock Academic Press, Waltham, MA, May 2011.

\bibitem{sorour10minimum}
S.~Sorour and S.~Valaee,
\newblock ``Minimum broadcast decoding delay for generalized instantly
  decodable network coding,''
\newblock {\em IEEE Global Telecommunications Conference (GLOBECOM 2010)}, pp.
  1--5, 2010.

\bibitem{C:BCDMin}
S.~Sorour and S.~Valaee,
\newblock ``On minimizing broadcast completion delay for instantly decodable
  network coding,''
\newblock in {\em Proc. IEEE International Conference on Communications
  (ICC'10)}, Cape Town, May 2010.

\bibitem{J:CDMin}
S.~Sorour and S.~Valaee,
\newblock ``Completion delay minimization for instantly decodable network
  codes,''
\newblock {\em submitted to IEEE/ACM Transactions on Networking},
\newblock http://arxiv.org/pdf/1201.4768v1.pdf.

\end{thebibliography}

\end{document}